\documentclass[a4paper,UKenglish,cleveref, autoref]{lipics-v2019}

\usepackage{xspace}

\newcommand{\com}[1]{}

\newcommand{\seclab}[1]{\label{sec:#1}}

\newcommand{\thmlab}[1]{{\label{theo:#1}}}
\newcommand{\thmref}[1]{Theorem~\ref{theo:#1}}

\newcommand{\lemlab}[1]{\label{lemma:#1}}
\newcommand{\lemref}[1]{Lemma~\ref{lemma:#1}}

\newcommand{\figlab}[1]{\label{fig:#1}}
\newcommand{\figref}[1]{Figure~\ref{fig:#1}}

\providecommand{\pth}[2][\!]{#1\left({#2}\right)}
\renewcommand{\O}[1]{\ensuremath{{\mathcal{O}\pth{#1}}}}
\newcommand{\Frechet}{Fr\'echet\xspace}
\newcommand{\etal}{\textit{e{}t~a{}l.}\xspace}
\providecommand{\eps}{{\varepsilon}}%
\newcommand{\distFr}[2]{\ensuremath{d_F\pth{#1,#2}}}
\newcommand{\distDF}[2]{\ensuremath{d_{dF}\pth{#1,#2}}}

\usepackage[ruled,vlined,linesnumbered]{algorithm2e}
\SetFuncSty{textsc}
\SetDataSty{textit}
\SetCommentSty{textrm}
\SetKwFunction{Extract}{ExtractFirstDigit}
\SetKwData{Forward}{forward}
\SetKw{True}{true}
\SetEndCharOfAlgoLine{}

\usepackage{microtype}

\title{On the complexity of the middle curve problem
} 

\titlerunning{On the complexity of the middle curve problem}

\author{Maike Buchin}{Ruhr-Universit\"at Bochum, Germany }{maike.buchin@rub.de}{}{}

\author{Nicole Funk}{Department of Computer
        Science, TU Dortmund, Germany }{nicole.funk@tu-dortmund.de}{}{}

\author{Amer Krivo\v{s}ija}{Department of Computer
        Science, TU Dortmund, Germany }{amer.krivosija@tu-dortmund.de}{}{}

\authorrunning{M. Buchin, N. Funk and A. Krivo\v{s}ija}

\Copyright{Maike Buchin, Nicole Funk and Amer Krivo\v{s}ija}



\relatedversion{}

\funding{A. Krivo\v{s}ija was supported by the German Science Foundation (DFG) Collaborative Research Center SFB 876 "Providing Information by Resource-Constrained Analysis", project A2.}

\EventEditors{Gill Barequet and Yusu Wang}
\EventNoEds{2}
\EventLongTitle{35th International Symposium on Computational Geometry (SoCG 2019)}
\EventShortTitle{SoCG 2019}
\EventAcronym{SoCG}
\EventYear{2019}
\EventDate{June 18--21, 2019}
\EventLocation{Portland, United~States}
\EventLogo{socg-logo}
\SeriesVolume{129}
\ArticleNo{0} 
\nolinenumbers 
\hideLIPIcs  

\begin{document}

\maketitle

\begin{abstract}
 For a set of curves, Ahn et al. \cite{ahnetal16} introduced the notion of a \emph{middle curve} and gave algorithms computing these with run time exponential in the number of curves. Here we study the computational complexity of this problem: we show that it is NP-complete and give approximation algorithms.
\end{abstract}

\section{Introduction}

Consider a group of birds migrating together. Several of these birds are GPS-tagged to analyze their behavior. The resulting data is a set of sequences of their positions. Such a sequence of data points can be interpreted as a polygonal curve. We want to represent the movement of the whole group, for instance to compare it to other groups or species. For this, we use a representative curve. Such a representative curve is also useful in other applications, such as the analysis of handwritten text or speech recognition.

There have been a few different approaches of defining such a representative curve. Buchin~\etal~\cite{median} defined the median level of curves as only using parts input curves, where the median can change directions where two input curves cross paths. Har-Peled and Raichel~\cite{hpr-fd} defined a mean curve, which can be chosen freely and minimizes the distance to the input curves. They gave an algorithm exponential in the number of curves for computing this.

Another approach is a version of the $(k,\ell)$-center problem, which asks for a set of $k$ center curves of complexity at most $\ell$ for which the distance of each input curve to its nearest center is minimized. In particular, the $(1,\ell)$-center problem asks for only one such center curve. The $(k,\ell)$-center problem for curves was first introduced by Driemel~\etal~\cite{dks16} and further analyzed by Buchin~\etal \cite{buchinsoda19} and Buchin~\etal \cite{bds-average}.

However, none of these representative curves use only actual data points of the GPS tracks. This could lead to the representative curves containing positions that the moving entities (e.g. birds) could not have visited.
As the data points in the input curves are more reliable Ahn \etal \cite{ahnetal16} defined the \textit{middle curve} to only use these points. For a more accurate representation of the original curves, Ahn \etal \cite{ahnetal16} define three variants of the middle curve. We use their definition of a middle curve in this paper.


\subparagraph*{Related work}
Ahn~\etal \cite{ahnetal16} presented algorithms for all three variants of the middle curve problems, whose running time is exponential in the number of input curves.
For several representative curve problems it is known that they are NP-hard, such as $(k,\ell)$-center \cite{buchinsoda19, dks16}, minimum enclosing ball \cite{buchinsoda19}, $(k,\ell)$-median \cite{dks16}, $1$-median under \Frechet and dynamic time warping distance \cite{bds-average,bulteau14}. 
Some problems are NP-hard even to approximate better than a constant factor, e.g. the $(k,\ell)$-center problem \cite{buchinsoda19}.
Similarly, Buchin \etal \cite{curvesimpli} showed, that assuming the Strong Exponential Time Hypothesis (SETH) the \Frechet distance of $k$ curves of complexity $n$ each cannot be computed significantly faster than $\O{n^k}$ time.

\subparagraph*{Our results}
We prove NP-completeness of the \textsc{Middle Curve} problem presented by Ahn \etal \cite{ahnetal16}. Next we define a parameterized version of the problem, and present a simple exact algorithm as well as an $(2+\eps)$-approximation algorithm for the parameterized problem.

\section{Preliminaries}

A polygonal curve $P$ is given by a sequence of vertices $\langle p_1, \dots, p_m \rangle$ with $p_i$ in $\mathbb{R}^d$, $1\leq i\leq m$, and  for $1 \leq i < m$ the pair of vertices $(p_i, p_{i+1})$ is connected by the straight line segment $\overline{p_i p_{i+1}}$. We call the number of vertices $m$ of the curve its \textit{complexity}. Let the input consist of $n$ polygonal curves $\mathcal{P} = \{ P_1, \dots, P_n\}$, each of complexity $m$.

\subparagraph*{\Frechet Distance}
We define the discrete \Frechet distance of two curves $P'=\langle p'_1, \ldots, p'_{m'}\rangle$ and $P''=\langle p''_1, \ldots, p''_{m''}\rangle$ as follows: we call a \emph{traversal} $T$ of $P'$ and $P''$ a sequence of pairs of indices $(i,j)$ of vertices $(p'_i,p''_j)\in P'\times P''$ such that
\begin{enumerate}
\item[i)] the traversal $T$ begins with $(1, 1)$ and ends with $(m', m'')$, and
\item[ii)] the pair $(i, j)$ of $T$ can be followed only by one of $(i+1, j)$, $(i, {j+1})$, or $({i+1}, {j+1})$.
\end{enumerate}
We note that every traversal is monotone. Denote $\mathcal{T}$ the set of all traversals $T$ of $P'$ and $P''$. The discrete \Frechet distance between $P'$ and $P''$ is defined as:
\begin{equation*}
\distDF{P'}{P''} = \min_{T \in \mathcal{T}} \max_{(i,j) \in T} \| p_i - q_j \|_2.
\end{equation*}
We call the set of pairs of vertices $(p',p'')\in P'\times P''$ that realize $\distDF{P'}{P''}$ a \emph{matching}, and say that these pairs of vertices are matched.

A related similarity measure is the continuous \Frechet distance. Let $\pi':[0,1]\rightarrow P'$ and $\pi'':[0,1]\rightarrow P''$ be two continuous functions on $[0,1]$ such that $\pi'(0)=p'_1$, $\pi'(1)=p'_{m'}$, $\pi''(0)=p''_1$, and $\pi''(1)=p''_{m''}$, and such that $\pi'$ and $\pi''$ are monotone on $P'$ and $P''$ respectively. Let $\mathcal{H}$ be the set of continuous and increasing functions $f:[0,1]\rightarrow [0,1]$ with $f(0)=1$ and $f(1)=1$. The continuous \Frechet distance between 
$P'$ and $P''$ is defined as:
\begin{equation*}
\distFr{P'}{P''} = \inf_{f \in \mathcal{H}} \max_{t \in [0,1]}  \|P'(\pi'(f(t))) - P''(\pi''(t))\|_2.
\end{equation*}
We can overload the notion, and say that the function mapping $P'$ and $P''$ that realizes $\distFr{P'}{P''}$ is a matching. 

Note that by definition the discrete \Frechet distance of $P'$ and $P''$ is an upper bound for the continuous \Frechet distance, as the traversal $T$ realising $\distDF{P'}{P''}$ can be extended into mapping between $P'$ and $P''$. Both $d_{DF}$ and $d_F$ are metrics.

\bigskip
\subparagraph*{Middle Curve}
Given a set of $n$ polygonal curves $\mathcal{P}$, a value $\delta \geq 0$, and a distance measure $\gamma$ for polygonal curves. We use $\gamma=d_{DF}$ as in \cite{ahnetal16}, for the continuous \Frechet distance $d_F$ the definitions hold verbatim. A \textbf{middle curve} at distance $\delta$ to $\mathcal{P}$ is a curve $M = \langle m_1, \dots, m_{\ell}\rangle$ with vertices $m_i \in \bigcup_{P_j \in \mathcal{P}} \bigcup_{p\in P_j} \lbrace p\rbrace$, $1\leq i\leq \ell$, s.t. $\max\{\distDF{M}{P_j} \colon P_j \in \mathcal{P}\} \leq \delta$ holds. 

If the vertices of a middle curve $M$ respect the order given by the curves of $\mathcal{P}$, then we call $M$ an \textbf{ordered middle curve}. Formally, for all $1\leq j\leq n$, if the vertex $m_i\in M$ is matched to $p_k\in P_j$ realizing $\distDF{M}{P_j}$, then for the vertices $m_{i'}\in M$, $i<i'$, it holds that $m_{i'}\in \left( \bigcup_{P_x \in \mathcal{P}\setminus P_j} \bigcup_{p\in P_x} \lbrace p\rbrace\right) \cup \left( \bigcup \lbrace p_{k'} : p_{k'}\in P_j, k'>k \rbrace \right)$. If the vertices of $M$ are matched to themselves in their original curves $P\in \mathcal{P}$ in the matching realizing $\distDF{M}{P}\leq\delta$, we have a \textbf{restricted middle curve}.
Note that an ordered middle curve is a middle curve, and a restricted middle curve is ordered as well.

We define the decision problem corresponding to finding such a curve. Given a set of polygonal curves $\mathcal{P} = \lbrace P_1, \ldots, P_n\rbrace$ and a $\delta \geq 0$ as parameters. \textsc{Unordered Middle Curve} problem returns \textsc{true} iff there exists a middle curve $M$ at distance $\delta$ to $\mathcal{P}$. The \textsc{Ordered Middle Curve} and \textsc{Restricted Middle Curve} returns \textsc{true} iff there exists an ordered and a restricted middle curve respectively at distance $\delta$ to $\mathcal{P}$.

Ahn \etal \cite{ahnetal16} presented dynamic programming algorithms for each variant of the middle curve problem. The running times of these algorithms for $n \geq 2$ curves of complexity at most $m$ are $O(m^{n} \log m)$ for the unordered case, $O(m^{2n})$ for the ordered case, and $O(m^n \log^n m)$ for the restricted middle curve case. All three cases have running time exponential in $n$, yielding the question if there is a lower bound for these problems.
In the following section we prove that the \textsc{Middle Curve} problem is NP-complete.

\section{NP-completeness}
\seclab{sec:proof}

The technique for the proof that all variants of the \textsc{Middle Curve} are NP-hard is based on the proof by Buchin \etal \cite{buchinsoda19} and Buchin, Driemel, and Struijs \cite{bds-average} for the NP-hardness of the Minimum enclosing ball and 1-median problems for curves under \Frechet distance. Their proof is a reduction from the \textsc{Shortest Common Supersequence} (\textsc{SCS}), which is known to be NP-hard \cite{SCS}. \textsc{SCS} problem gets as input a set $\mathcal{S} = \{S_1, \ldots, S_n\}$ of $n$ sequences over a binary alphabet $\Sigma = \{A, B\}$ and $t\in \mathbb{N}$. \textsc{SCS} returns \textsc{true} iff there exists a sequence $S^*$ of length at most $t$, that is a supersequence of all sequences in $\mathcal{S}$.

Our NP-hardness proof differs from the proof of \cite{buchinsoda19,bds-average} in three aspects. First, the mapping of the characters of the sequence is extended by additional points. Second, in order to validate all three variants of our problem, the conditions of the restricted middle curve have to be fulfilled, i.e. each vertex has to be matched to itself. Third, our representative curve is limited to the vertices of the input curves. Due to the hierarchy of the middle curve problems we show the reductions from \textsc{SCS} to the \textsc{Restricted Middle Curve}, and from \textsc{Unordered Middle Curve} to \textsc{SCS}. 

Given set $\mathcal{S}=\{S_1, \ldots, S_n\}$ of sequences over $\Sigma = \{A,B\}$ and $t\in\mathbb{N}$, defining a \textsc{SCS} instance that returns \textsc{true}, we construct for each sequence $S_i\in \mathcal{S}$ a polygonal curve in one-dimensional space, and therewith a \textsc{Middle Curve} instance. We use the following points in $\mathbb{R}$:
\begin{equation}
	\begin{aligned}
		p_{-3} = -3, & & p_{3} = 3, & & p_{-1} = -1, & & p_{1} = 1, \text{ and} \\
		p_0 = 0, & & p_2 = 2, & & p_{-2} = -2.
	\end{aligned}
\end{equation}
Each character in a sequence $S_i\in \mathcal{S}$ is mapped to a curve over $\mathbb{R}$ as follows:
\begin{equation}
	\begin{aligned}
		A &\to & p_0 (p_{-1} p_{1})^t p_{-2} p_{-3} p_{-2} (p_{1} p_{-1})^t p_0 ,\\
		B &\to & p_0 (p_{1} p_{-1})^t p_{2} p_{3} p_{2} (p_{-1} p_{1})^t p_0.
	\end{aligned}
\end{equation}
The curve $\eta(S_i)$ representing the sequence $S_i \in \mathcal{S}$ is constructed by concatenating the curves resulting from each character's mapping. The set of all resulting curves is denoted by $G = \{\eta(S_i) \colon S_i \in \mathcal{S}\}$. We call the subcurves $p_{-2} p_{-3} p_{-2}$ and $p_{2} p_{3} p_{2} $ \textit{letter $A$} and \textit{letter $B$ gadgets} respectively, and the subcurves between two letter gadgets (or at the beginning and at the end of curves) consisting of $p_{-1}$, $p_{1}$, and $p_0$ \textit{buffer gadgets}. 

We define the set $I_t = \{(a,b) \in \mathbb{N}^2 \colon a,b \geq 0, a+b = t \}$. A pair $(a,b)\in I_t$ represents the number of $A$'s and $B$'s in a possible supersequence of length $t$. For some $(a,b) \in I_t$  we construct the curves $A^a$ and $B^b$ in $\mathbb{R}$ with
\begin{equation}
	\begin{aligned}
	A^a &= p_{1}(p_{-3} p_{1})^a \\
	B^b &= p_{-1}(p_{3} p_{-1})^b.
	\end{aligned}
\end{equation}
We use these curves to construct the \textsc{Middle Curve} instance $(G \cup \{A^a, B^b\}, 1)$ for a pair $(a,b) \in I_t$. 
We prove that the \textsc{SCS} instance $(\mathcal{S},t)$ returns \textsc{true} if and only if there exists a pair $(a,b) \in I_t$ such that $(G \cup \{A^a, B^b\}, 1)$ is a \textsc{Middle Curve} instance that returns \textsc{true}. An example for this construction is given in \figref{fig:reduction}.

\begin{figure}
\hspace{-3.5mm}
\begin{tabular}{cc}
\includegraphics[page=2, width = .48\textwidth]{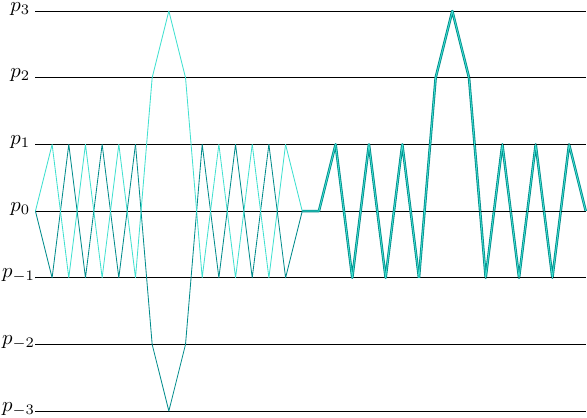}&\includegraphics[page=3, width = .48\textwidth]{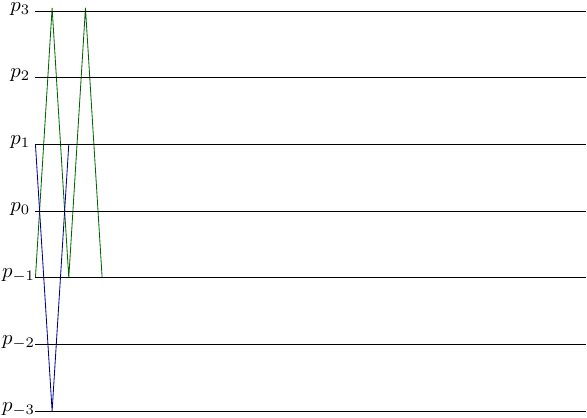}\\
\textbf{(a)} Curves $G = \{\eta(AB) , \eta(BB)\}$
&\textbf{(b)} Curves $A^1$ and $B^2$
\end{tabular}
\vspace{2mm}

\hspace{-3.5mm}
\begin{tabular}{cc}
\includegraphics[page=5, width = .48\textwidth]{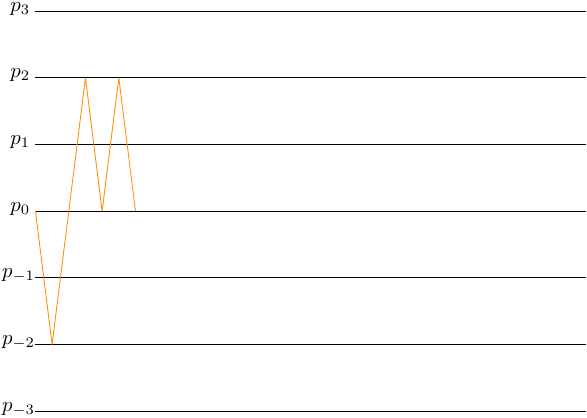}&\includegraphics[page=6, width = .48\textwidth]{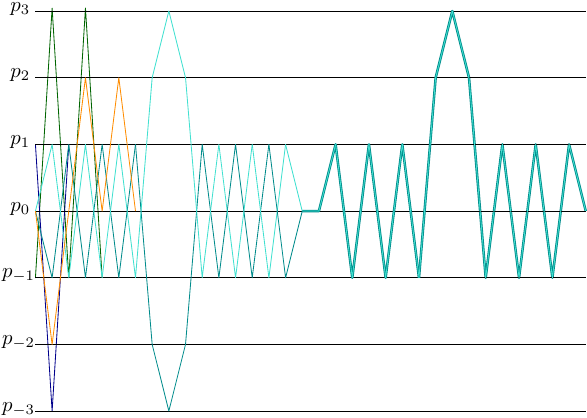}\\
\textbf{(c)} Curve $M$ representing the sequence $ABB$
&\textbf{(d)} Set $G \cup \{A^1, B^2\} \cup \{M\}$
\end{tabular}
\caption{Construction of the \textsc{Middle Curve} instance $(G \cup \{A^1, B^2\}, 1)$ for the \textsc{SCS} instance $(\{AB, BB\} , 3)$}
\figlab{fig:reduction}
\end{figure}

We consider the discrete \Frechet distance case first.

\begin{lemma}
\lemlab{Lemma:hin}
If $(\mathcal{S},t)$ is a \textsc{SCS} instance returning \textsc{true}, then there exists a pair $(a,b) \in I_t$ such that $(G \cup \{A^a, B^b\}, 1)$ is a \textsc{Restricted Middle Curve} instance for the discrete \Frechet distance that returns \textsc{true}.
\end{lemma}
\begin{proof}
If $(\mathcal{S},t)$ is a \textsc{SCS} instance returning \textsc{true}, then there exists a supersequence of the curves in $\mathcal{S}$ with length at most $t$. Let $S^*$ be this supersequence with letters $s^*_i$, for $i \in \{1,\ldots,t\}$.

We construct a curve $M = \langle m_1, \ldots, m_{2t+1} \rangle$ using vertices of the curves in $G$, such that $M$ represents $S^*$. The vertex $m_j$ for $ j \in \{1, \ldots, 2t+1\}$ is defined as:
\begin{equation*}
	m_j = \begin{cases}
		p_0 & j \text{ is odd,} \\
		p_{-2} & j \text{ is even and } s^*_{j/2} = A,\\
		p_{2} & j \text{ is even and } s^*_{j/2} = B.
	    \end{cases}
\end{equation*}
The vertices with even indices in $M$ represent the characters in $S^*$ while the vertices with odd indices act as a buffer between them.
For every $S_i \in \mathcal{S}$ the curve $\eta(S_i) \in G$ is constructed. We construct a traversal between $\eta(S_i)$ and $M$, that realizes $\distDF{\eta(S_i)}{M}$ that is at most 1. Since $S_i$ is a subsequence of $S^*$, we proceed as follows: we match the first $p_0\in \eta(S_i)$ to the $p_0 \in M$. Then, as long as there are letters in $S^*$, do:
\begin{description}
\item If the current letter in $S_i$ and $S^*$ is the same, then match the next buffer gadget (and the possible rest of the previous buffer gadget) in $\eta(S_i)$ to $p_0\in M$, then match the letter $A$ gadget to $p_{-2}$ or the letter $B$ gadget to $p_2$ respectively. Move to the next letter in both $S_i$ and $S^*$.
\item If the current letter in $S_i$ and $S^*$ differ, or there are no more letters in $S_i$, then we have the following cases depending on the letters in $S^*$:
\begin{itemize}
\item last letter in $S^*$ was $A$ and the current one is $A$: match $p_1\in \eta(S_i)$ to $p_0\in M$, and match $p_{-1}\in \eta(S_i)$ to $p_{-2}\in M$;
\item last letter in $S^*$ was $A$ and the current one is $B$: match $p_1\in \eta(S_i)$ to $p_0\in M$, and match the same $p_{1}\in \eta(S_i)$ to $p_{2}\in M$;
\item last letter in $S^*$ was $B$ and the current one is $A$: match (already seen) $p_1\in \eta(S_i)$ to $p_0\in M$, and match $p_{-1}\in \eta(S_i)$ to $p_{-2}\in M$;
\item last letter in $S^*$ was $B$ and the current one is $B$: match $p_{-1}\in \eta(S_i)$ to $p_0\in M$, and match $p_{1}\in \eta(S_i)$ to $p_{2}\in M$.
\end{itemize}
In each case move to the next letter in $S^*$. Notice that this case can happen at most $t$ times, and that each letter uses at most 2 further vertices in the current buffer gadget, thus the $t$ iterations of the pair $p_1p_{-1}$ in a buffer gadget suffice.
\end{description}
We conclude with matching the rest of the last buffer gadget in $\eta(S_i)$ to $p_0\in M$. 
Notice that the vertices in $M$ are matched to themselves in $\eta(S_i)$, while vertices at $p_0$, $p_2$, or $p_{-2}$ are matched to a vertex at the same position in $M$, thus the conditions for a restricted middle curve are met. The distance between the matched points is at most $1$, thus we have $\distDF{\eta(S_i)}{M} \leq 1$, for all $S_i\in \mathcal{S}$.

Set $a$ and $b$ to the number of occurrences of $A$ and $B$ in $S^*$ respectively, therefore it is $a+b = t$.
Per definition $M$ contains $p_{-2}$ exactly $a$ times, thus we can match these $p_{-2}$ to the vertices $p_{-3}\in A^a$, while the remaining vertices $p_0, p_2\in M$ can be matched to the vertices $p_1\in A^a$, respecting the order of the vertices on $A^a$ and $M$. Analogously the $b$ vertices $p_2\in M$ can be matched to the vertices $p_3\in B^b$, and the vertices $p_0,p_{-2}\in M$ can be matched to $p_{-1}\in B^b$. Therefore it holds that $\distDF{A^a}{M}\leq 1$ and $\distDF{B^b}{M}\leq 1$.
So $M$ is a restricted middle curve of $G \cup \{A^a, B^b\}$ at distance 1, as claimed.
\end{proof}

\begin{lemma}
\lemlab{Lemma:rueck}
If there exists a pair $(a,b) \in I_t$ such that $(G \cup \{A^a, B^b\}, 1)$ is an \textsc{Unordered Middle Curve} instance for the discrete \Frechet distance that returns \textsc{true}, then $(\mathcal{S},t)$ is a \textsc{SCS} instance that returns \textsc{true}.
\end{lemma}
\begin{proof}
Given a pair $(a,b) \in I_t$, let $M$ be an unordered middle curve  of the set $(G \cup \{A^a, B^b\}$ at distance 1. We construct a sequence that represents the curve $M$ and prove that every $S_i \in S$ is a subsequence of this sequence.

Since $\distDF{A^a}{M}\leq 1$, we observe a matching between $A^a$ and $M$ that realizes $\distDF{A^a}{M}$. Since $A^a$ consists only of vertices $p_{1}$ and $p_{-3}$, and there cannot exist a point in $\mathbb{R}$ with distance at most 1 to both of these vertices, every vertex in $M$ can only be matched to one vertex in $A^a$. Since for every two vertices $p_{-3}$ in $A^a$ there is a $p_1$ vertex between them in $A^a$, a vertex in $M$ can be matched to at most one $p_{-3}$ in $A^a$. The same holds for the vertices $p_1$. Thus every vertex in $M$ is matched to exactly one vertex in $A^a$. Analogously every vertex in $M$ is matched to exactly one vertex in $B^b$.

Thus we can partition the vertices of $M$ into $2a+1$ subsets $M^a_i$, $i\in \{1,\ldots,2a+1\}$, where all vertices within one subset $M^a_i$ are matched to the $i$-th vertex in $A^a$ (in the matching realizing $\distDF{A^a}{M}$). Analogously we can partition the vertices of $M$ into $2b+1$ subsets $M^b_j$, $j\in\{1,\ldots,2b+1\}$ (using the matching realizing $\distDF{B^b}{M}$). We combine these partitions into one. We call the subsets $M^a_i$ that represent $p_{-3}\in A^a$ the $A$-subsets, and the subsets $M^b_j$ that represent $p_3\in B^b$ the $B$-subsets. 

We note that there cannot exist a vertex in $M$ that is simultaneously in some $A$-subset and some $B$-subset, otherwise it would be at distance at most 1 to both $p_3$ and $p_{-3}$. We take over the $A$- and $B$-subsets into the new partition (and call them letter subsets). By construction there are $a+b=t$ letter subsets. The remaining vertices in $M$ -- either before the first letter subset along $M$, or after the last letter subset, or between two letter subsets form the pairwise disjunct buffer subsets, and thus together with letter subsets define a partition of the vertices of $M$. There can be at most $t+1$ buffer subsets, thus there are at most $2t+1$ subsets in the constructed partition of the vertices of $M$.
\figref{fig:proof} shows an example of such a partition.

\begin{figure}
\hspace{-3.5mm}
\begin{tabular}{cc}
\includegraphics[page=4, width = .48\textwidth]{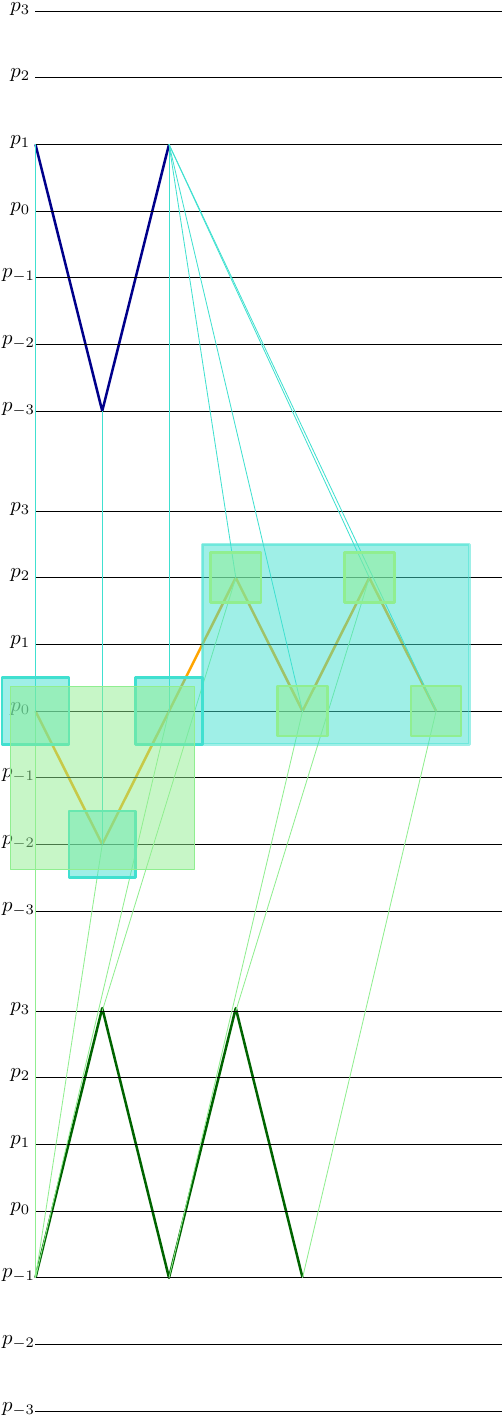}&\includegraphics[page=3, width = .48\textwidth]{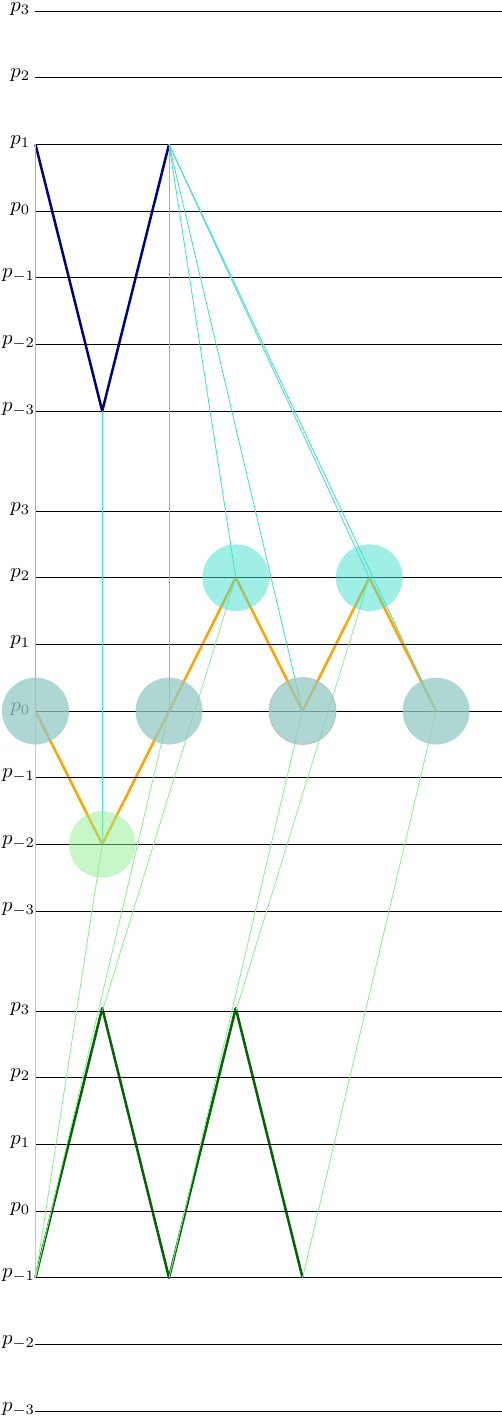}\\
\textbf{(a)} Individual partitions of $M$ 
&\textbf{(b)} Combined partition of $M$
\end{tabular}
\caption{An example of possible matchings between $M$ and $A^1$ or $B^2$, given in 
\figref{fig:reduction}. \\\textbf{(a)} The individual matchings and partition of $M$ based on $A^1$ (light blue lines and light blue boxes) and on $B^2$ (light green lines and light green boxes) respectively. \\\textbf{(b)} The combined partition of $M$ is represented by the circles around vertices (light blue -- A-parts, light green -- B-parts, gray -- buffer parts.}
\figlab{fig:proof}
\end{figure}

The sequence $S^*$ can be constructed using the constructed partition of $M$, by replacing the $A$-subsets with the letter $A$, and the $B$-subsets with the letter $B$. The buffer subsets are simply omitted. The sequence $S^*$ has length $t$. We need to prove that $S^*$ is a supersequence of all sequences in $\mathcal{S}$.

Let for some $S_i\in\mathcal{S}$ be $\eta(S_i) \in G$ its representing curve. As $M$ is a middle curve of $G \cup \{A^i, B^j\}$ at distance 1, there exists a matching of $\eta(S_i)$ and $M$ that realizes $\distDF{\eta(S_i)}{M}\leq 1$. In this matching a vertex in one $A$-subset (of the partition of the vertices of $M$) cannot be matched to two vertices in different letter gadgets (in $\eta(S_i)$), since the buffer gadget separating two letter gadgets contains the vertex $p_{1}$, which cannot be matched to a vertex in a $A$-subset with distance at most $1$. Analogously, a vertex in one $B$-subset cannot be matched to vertices in two different letter gadgets. 

Each letter $A$ gadget in $\eta(S_i)$ contains vertex $p_{-3}$ which has to be matched to a vertex in an $A$-subset (otherwise by construction it would be at distance at most 1 to $p_{1}$). Analogously, each letter $B$ gadget in $\eta(S_i)$ contains vertex $p_3$ which has to be matched to a vertex in a $B$-subset. Thus each letter gadget in $\eta(S_i)$ corresponds one-to-one to a letter subset in $M$, and the sequence of letter gadgets in $\eta(S_i)$ corresponds to the sequence of letter subsets in $M$. Therefore $S_i$ is a subsequence of $S^*$, as claimed.
\end{proof}

\lemref{Lemma:hin} and \lemref{Lemma:rueck} imply the following theorem for the discrete \Frechet distance.

\begin{theorem}
\thmlab{theorem:disc}
Every variant of \textsc{Middle Curve} problem for the discrete and the continuous \Frechet distance is NP-hard.
\end{theorem}
\begin{proof}[Proof of \thmref{theorem:disc} for the discrete \Frechet distance]
We reduce from the \textsc{SCS} problem, which is known to be NP-hard. \lemref{Lemma:hin} and \lemref{Lemma:rueck} show that this construction is a viable reduction.

Given the \textsc{SCS} instance $(\mathcal{S},t)$, the \textsc{Middle Curve} instance $(G \cup \{A^a, B^b\}, 1)$ for a pair $(a,b) \in I_t$ can be constructed in a time linear in the input size. As the number of possible pairs $(a,b) \in I_t$ for a given supersequence of length $t$ is linear in $t$, the number of different \textsc{Middle Curve} instances is also linear in $t$.
Thus the reduction can be computed in a time polynomial in the input size of the \textsc{SCS} instance.
\end{proof}

Like the proof of Buchin \etal \cite{buchinsoda19}, the shown reduction for the discrete \Frechet distance can be adopted to the continuous \Frechet distance to prove \thmref{theorem:disc} in that case too. \lemref{Lemma:cont:hin} and \lemref{Lemma:cont:rueck} take place of \lemref{Lemma:hin} and \lemref{Lemma:rueck} respectively. The rest of the proof is taken verbatim.

\begin{lemma}
\lemlab{Lemma:cont:hin}
If $(\mathcal{S},t)$ is an instance of the \textsc{SCS} that returns \textsc{true}, then there exists a pair $(a,b) \in I_t$ such that $(G \cup \{A^a, B^b\}, 1)$ is a \textsc{Restricted Middle Curve} instance for the continuous \Frechet distance that returns \textsc{true}.
\end{lemma}
\begin{proof}
Given the \textsc{SCS} instance $(\mathcal{S},t)$ returning \textsc{true}, \lemref{Lemma:hin} implies that there exists $(a,b) \in I_t$, such that $\distDF{g}{M} \leq 1$ for $g \in G \cup \{A^a, B^b\}$, and the restricted middle curve $M$ constructed in its proof. Since the discrete \Frechet distance is an upper bound for the continuous \Frechet distance, we have $\distFr{g}{M} \leq \distDF{g}{M} \leq 1$ for all $g \in G \cup \{A^a, B^b\}$. This means that $M$ is also a restricted middle curve for the continuous \Frechet distance.
\end{proof}

\begin{lemma}
\lemlab{Lemma:cont:rueck}
If there exists a pair $(a,b) \in I_t$ such that $(G \cup \{A^a, B^b\}, 1)$ is an \textsc{Unordered Middle Curve} instance for the continuous \Frechet distance that returns \textsc{true}, then $(\mathcal{S},t)$ is a \textsc{SCS} instance that returns \textsc{true}.
\end{lemma}
\begin{proof}
Given a pair $(a,b)\in I_t$, let $M$ be an unordered middle curve of the set $G\cup \{A^a, B^b\}$ at distance 1. We adapt the proof of \lemref{Lemma:rueck} to the continuous case.

Since $\distFr{A^a}{M}\leq 1$, there has to be a point $q_a$ on the curve $M$ that is at distance at most 1 to the vertex $p_{-3}\in A^a$, for each such a vertex. Thus $q_a\in [-2,-3]$. But since $\distFr{B^b}{M}\leq 1$, there has to be a point on $B^b$ at distance at most 1 to $q_a$, thus such a point is in $[-2,-1]$. Since all points on $B^b$ lie in $[-1,3]$, it implies that that point has to be exactly at $-1$, thus $q_a=p_{-2}$. We call that point an $A$-subset of $M$. It is possible that the curve $M$ contains several consecutive vertices at $p_{-2}$, and in that case the whole subcurve defined by such vertices is an $A$-subset of $M$. Analogously, we conclude that for each $p_3\in B^b$ there is a point $p_2\in M$, and call it a $B$-subset of $M$.

As in \lemref{Lemma:rueck} we partition the curve $M$ into $2a+1$ (respectively $2b+1$) subcurves $M^a_i$, $i\in \{1,\ldots, 2a+1\}$ (resp. $M^b_j$, $j\in \{1,\ldots, 2b+1\}$), where the subcurves with even indices are $A$-subsets (resp. $B$-subsets) of $M$, and the rest of the curve $M$ defines the subcurves with odd indices. Again, we combine these two partitions of $M$ into one, since no vertex on $M$ can be in both $A$- and $B$-subsets. The sequence $S^*$ is constructed by replacing each letter subset in $M$ with the corresponding letter.

The rest of the proof of \lemref{Lemma:rueck} follows, since for each $S_i\in \mathcal{S}$ and for the matching that realizes $\distFr{\eta(S_i)}{M}\leq 1$ it holds that a vertex in one $A$-subset in $M$ cannot be matched to the vertex $p_{-3}$ in two different letter $B$ gadgets in $\eta(S_i)$, and each vertex $p_{-3}\in \eta(S_i)$ has to be matched to a vertex in an $A$-subset. The analogous claim can be made for $B$-subsets. There is a one-to-one correspondence between the letter gadgets in $\eta(S_i)$ and the letter subsets in $M$, thus $S_i$ is a subsequence of $S^*$.
%
%
%
\end{proof}

\bigskip

Using \thmref{theorem:disc}, we can now prove the NP-completeness of the \textsc{Middle Curve} decision problem. Given a \textsc{Middle Curve} instance $(\mathcal{P}, \delta)$ with $\mathcal{P}$ containing $n$ curves of complexity $m$, we guess non-deterministically a middle curve $M$ of complexity $\ell$. We can decide whether the \Frechet distance between $M$ and a curve $P \in \mathcal{P}$ is at most $\delta$ in $\O{ m\ell}$ time using the algorithm by Alt and Godau \cite{altgodau} for the continuous, and by Eiter and Mannila \cite{eiter1994computing} for the discrete \Frechet distance. We note that the algorithm by Alt and Godau \cite{altgodau} has to be modified a bit, as it uses a random access machine instead of a Turing machine, as this allows the computation of square roots in constant time. But comparing the distances is possible by comparing the squares of the square roots, thus this results in a non-deterministic $\O{n m\ell }$-time algorithm for the \textsc{Unordered Middle Curve} problem.

In order to decide the  \textsc{Ordered Middle Curve} problem, it is necessary to compare the middle curve to the input curves, which is possible in $O(nm)$ time. For the restricted \textsc{Restricted Middle Curve} problem the matching corresponding to the Frechet distance $\leq \delta$ has to be known. This matching is a result of the decision algorithm by Alt and Godau \cite{altgodau}. Given this matching it can be checked in $\O{m+\ell}$ time if a vertex is matched to itself. This yields the following theorem.

\begin{theorem}
\thmlab{theorem:complete}
Every variant of the \textsc{Middle Curve} problem for the discrete or continuous \Frechet distance is NP-complete.
\end{theorem}


If the \textsc{SCS} problem is parameterized by the number of input sequences $n$, it is known to be W[1]-hard \cite{bulteau14}.
In our reduction from \textsc{SCS} the number of input curves in the constructed \textsc{Middle Curve} instance 
is $n+2$.
Thus the shown reduction is also a parameterized reduction from \textsc{SCS} with the parameter $n$ to the \textsc{Middle Curve} problem parameterized by the number of input curves, yielding  the following theorem.

\begin{theorem}
Every variant of the \textsc{Middle Curve} problem for the discrete and continuous \Frechet distance parameterized by the number of input curves $n$ is W[1]-hard.
\end{theorem}

\section{Approximation algorithm}

A different way of parameterizing the \textsc{Middle Curve} problem is to use the complexity of the middle curve.
Given a set of polygonal curves $\mathcal{P}$, a $\delta \geq 0$, and a parameter $\ell \in \mathbb{N}$. We define the \textsc{parameterized middle curve} decision problems, that return \textsc{true} iff a middle curve of complexity $\leq \ell$ with corresponding conditions exists (for each of the three variants).

It is clear that there exists a simple brute force optimization algorithm for the \textsc{Parameterized Middle Curve} instance $(\mathcal{P},\delta, \ell)$, that tests all $\ell$-tuples of the vertices from the curves in $\mathcal{P}$ in $\O{(m n)^\ell m \ell \log m \ell}$. This holds for all three versions of the problem.

We want to give an approximation algorithm for \textsc{Parameterized Middle Curve} optimization problem for the discrete Frechet distance. For this we use an approximation of the $(k,\ell)$-center optimization problem on curves. The $(k,\ell)$-center problem for curves was introduced by Driemel \etal \cite{dks16}. Given a set $\mathcal{P} = \{P_1, \dots, P_n\}$ of polygonal curves of complexity at most $m$, it looks for a set of curves $\mathcal{C} = \{ C_1, \dots C_k \}$, each of complexity at most $\ell$, that minimizes $\max_{P \in \mathcal{P}} \min_{i=1}^k \gamma(C_i, P)$ for a distance measure $\gamma$. The unordered \textsc{Parameterized Middle Curve} optimization problem is a $(1,\ell)$-center problem, where the curve $C_1$ is limited to vertices from the input curves and the distance measure $\gamma$ is a variant of the \Frechet distance.

Given a set $\mathcal{P}$ of $n$ curves of complexity $m$ in $\mathbb{R}^d$, let $C$ be the $(1,\ell)$-center curve returned by some $\alpha$-approximation algorithm for the discrete \Frechet distance. Let $\delta=\max_{P\in \mathcal{P}} \distDF{C}{P}$. We construct $d$-dimensional balls centered at vertices of the curve $C$ with radius $\delta$. It holds that $\distDF{C}{P}\leq \delta$, $\forall P\in \mathcal{P}$, thus in each ball centered at the vertices of $C$ there has to be a vertex of each curve from $\mathcal{P}$. We choose at random one vertex from each of the $\ell$ balls, and connect them with line segments in the order of the vertices along $C$. We denote the curve we got with $M$, and claim that it is a good approximation of an unordered parameterized middle curve. See \figref{fig:approx} for an illustration of the algorithm.  

\begin{figure}
\centering
\includegraphics[page=8, width = .7\textwidth]{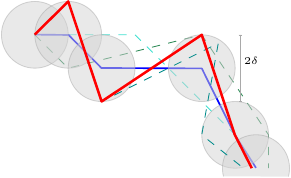}
\caption{Illustration of the approximation algorithm. The input curves are dashed and in shades of green, while the $(1,\ell)$-center approximation with distance $\delta$ is the full purple curve. The constructed middle curve is the red fat curve.}
\figlab{fig:approx}
\end{figure}

Let $C^*$ be an optimal $(1,\ell)$-center curve (for the discrete \Frechet distance) for the given input set $\mathcal{P}$. 
Let $\delta^* = \max_{P \in \mathcal P} \distDF{C^*}{P}$. It holds that $\delta \leq \alpha \delta^*$. 
For each $P\in\mathcal{P}$ and each vertex of $P$, there is a vertex in $M$, that is at distance at most $2\delta$ (diameter of the ball both of them lie in). Thus there is a traversal of $P$ and $M$ with pairwise distance of the vertices at most $2\delta$, implying $\distDF{M}{P}\leq 2\delta$. We have 
$\distDF{M}{P} \leq 2\delta \leq 2 \alpha \delta^*$.

Let the optimal parameterized middle curve with complexity $\ell$ be $M^*$. By definition it holds that $\delta^* = \max_{P\in \mathcal{P}} \distDF{C^*}{P} \leq \max_{P\in \mathcal{P}} \distDF{M^*}{P}$. Thus 
\[ \distDF{M}{P}\leq 2\alpha \max_{P\in \mathcal{P}} \distDF{M^*}{P},\] 
and $M$ is a 2-approximation to the optimal parameterized middle curve. This implies:

\begin{lemma}
\lemlab{approx:param:1}
Given a set of $n$ curves $\mathcal{P}$ each with complexity at most $m$, a $\delta > 0$ and an $\alpha$-approximation algorithm for $(k,\ell)$-center with running time $T$, we can compute a $2\alpha$-approximation of the \textsc{Parameterized Middle Curve} optimization problem for discrete \Frechet distance in $O(\ell mn + T)$ time.
\end{lemma}

Plugging the $(1+\eps)$-approximation algorithm of Buchin \etal \cite{bds-average} for $(k,\ell)$-center for discrete \Frechet distance into \lemref{approx:param:1}, we get

\begin{theorem}
Given a set of $n$ curves $\mathcal{P}$ each of complexity at most $m$, and a $\delta > 0$, we can compute a $(2 + \eps)$-approximation of the \textsc{Parameterized Middle Curve} optimization problem for discrete \Frechet distance in $\O{((c\ell)^{\ell}+\log(\ell+n))\ell mn}$ time, with $c = (\frac{4 \sqrt{d}}{\eps} + 1)^d$.
\end{theorem}

\section{Conclusion}

We showed that the \textsc{Middle Curve} problem is NP-complete and gave a $(2+\eps)$-approximation for the \textsc{Parameterized Middle Curve} problem, parameterized in the complexity of the middle curve. It would be interesting to gain further insight into the complexity of the parameterized problem. Fixing the parameter in the brute-force algorithm gives an XP-algorithm, however it remains open whether \textsc{Parameterized Middle Curve} is in FPT. 

\bigskip

\hspace{-5.5mm}\textbf{Acknowledgements.} This work is based on the student research project by the second author Nicole Funk.

\bibliography{middlecurveEuroCG}

\end{document}